\documentclass[letterpaper, 10 pt, conference]{ieeeconf}  % Comment this line out
\IEEEoverridecommandlockouts         % This command is only
\usepackage{balance} %even columns on last page

\usepackage{xcolor,calc}
\usepackage[pdftex]{graphicx}
\usepackage{amsmath} % assumes amsmath package installed
\usepackage{amssymb}  % assumes amsmath package installed
\usepackage[noadjust]{cite}
\usepackage{color}
\usepackage{enumerate}
\usepackage{gensymb}
\usepackage{subfigure}
\usepackage{multirow}
\usepackage{mathtools}
\usepackage{booktabs}
\usepackage{epstopdf}
\usepackage[linesnumbered,ruled,vlined]{algorithm2e}

\usepackage[normalem]{ulem}%%%%%%%%%

\usepackage{url}
\usepackage{bm}
\usepackage{caption}

\usepackage{booktabs}

\newtheorem{theorem}{Theorem}
\newtheorem{remark}{Remark}
\newtheorem{proposition}{Proposition}

\newtheorem{lemma}{Lemma}
\newtheorem{corollary}{Corollary}
\SetKwInput{KwInit}{Initialization}                % Set the Input
\SetKwInput{KwInput}{Input}
\SetKwInput{KwOutput}{Output}
\makeatletter
\newsavebox\myboxA
\newsavebox\myboxB
\newlength\mylenA

\newcommand*\xoverline[2][0.75]{%
    \sbox{\myboxA}{$\m@th#2$}%
    \setbox\myboxB\null% Phantom box
    \ht\myboxB=\ht\myboxA%
    \dp\myboxB=\dp\myboxA%
    \wd\myboxB=#1\wd\myboxA% Scale phantom
    \sbox\myboxB{$\m@th\overline{\copy\myboxB}$}%  Overlined phantom
    \setlength\mylenA{\the\wd\myboxA}%   calc width diff
    \addtolength\mylenA{-\the\wd\myboxB}%
    \ifdim\wd\myboxB<\wd\myboxA%
       \rlap{\hskip 0.5\mylenA\usebox\myboxB}{\usebox\myboxA}%
    \else
        \hskip -0.5\mylenA\rlap{\usebox\myboxA}{\hskip 0.5\mylenA\usebox\myboxB}%
    \fi}
\makeatother

\usepackage{tabularx}
\usepackage{booktabs}

\newtheorem{definition}{Definition}
\newtheorem{assumption}{Assumption}

\makeatletter
\let\NAT@parse\undefined
\makeatother
\usepackage{hyperref}

\title{\LARGE \bf Modeling of Dynamical Systems via Successive Graph Approximations}

\author{Siddharth H. Nair, Monimoy Bujarbaruah, and Francesco Borrelli  % <-this % stops a space
\thanks{The authors are with the MPC Lab, UC Berkeley, USA; E-mails: \tt\scriptsize{\{siddharth\_nair, monimoyb, fborrelli\}@berkeley.edu.}} %
}

\begin{document}

\maketitle
  \thispagestyle{empty}
\pagestyle{empty}

%%%%%%%%%%%%%%%%%%%%%%%%%%%%%%%%%%%%%%%%%%%%%%%%%%%%%%%%%%%%%%%%%%%%%%%%%%%%%%%%
\begin{abstract}
In this work, we propose a non-parametric technique for online modeling of systems with unknown nonlinear Lipschitz dynamics. The key idea is to successively utilize measurements to approximate the \emph{graph} of the state-update function using envelopes described by quadratic constraints. The proposed approach is then demonstrated on two  control applications: $(i)$ computation of tractable bounds for un-modelled dynamics, and $(ii)$ computation of positive invariant sets. We further highlight the efficacy of the proposed approach via a detailed numerical example. 
\end{abstract}

%%%%%%%%%%%%%%%%%%%%%%%%%%%%%%%%%%%%%%%%%%%%%%%%%%%%%%%%%%%%%%%%%%%%%%%%%%%%%%%%

\section{Introduction}\label{sec:intro}
Characterization of system model and associated uncertainties is of paramount importance while dealing with autonomous systems. In recent times, as data-driven decision making and control becomes ubiquitous \cite{recht2018tour, rosolia2018data}, system identification methods are being integrated with control algorithms for control of uncertain dynamical systems. In computer science community, data driven reinforcement learning algorithms \cite{bertsekas1989adaptive, watkins1992q} have been extensively utilized for policy and value function learning of uncertain systems. In control theory, if the actual model of a system is unknown, adaptive control \cite{krstic1995nonlinear, sastry2011adaptive} strategies have been applied for simultaneous system identification and control.  
Techniques for system modelling and identification have been traditionally rooted in statistics and data sciences \cite{friedman2001elements, hastie2017generalized}. Statistical models that describe observed data, can be classified into parametric \cite{hothorn2008simultaneous}, non-parametric and semi-parametric \cite{hardle2012nonparametric} models.
\par
Parametric models assume a model structure a priori, based on the application and domain expertise of the designer.
In almost all of classical adaptive control methods, parametric models are learned from data in terms of point estimates, and asymptotic convergence of such estimates are proven under persistence of excitation \cite{green1986persistence} conditions. 
The concept of online model learning and adaptation has been extended to systems under constraints as well, after obtaining a set or a confidence interval containing possible realizations of the system model. 
% Both the case of parametric and non-parametric model uncertainties have been looked into in this regard for suitably designing control algorithms. 
% In \cite{dean2018safely}, linear time invariant system dynamics matrices and the confidence intervals \cite{simchowitz2018learning} are learned using Ordinary Least Squares regression.  For solving such robust infinite horizon problems, the field of Adaptive MPC has caught on in recent times. In \cite{tanaskovic2014adaptive, bujarbaruahAdapFIR, bujarbaruahAdapCDC18, lorenzen2017adaptive, bertsekas1971recursive}, Set Membership Method \cite{bertsekas1971recursive} based approaches are used to obtain the sets containing all possible model uncertainty. These sets are modified as more data becomes available. In \cite{zhu2015adaptive} parametric model uncertainty sets are constructed from the sublevel sets of a Lyapunov function, which is updated online with data. 
Gaussian Mixture Modeling (GMM) \cite{ghahramani1999learning, kalouptsidis2011adaptive} has also been used to identify unknown system parameters, via an expectation maximization algorithm. 

% Although robust satisfaction of constraints is ensured with such set based parametric uncertainty modeling, but computational tractability can be obtained only for linear systems, due to the need for computation of robust invariant \cite{kolmanovsky1998theory} terminal sets \cite{borrelli2017predictive} for the MPC problem. For nonlinear uncertain systems subject to constraints, non-parametric modeling of uncertainty is therefore used to obtain lesser conservatism, at the expense of robust guarantees. 
\par
However, parametric models are restricted only to specified forms of function classes, and so to widen the richness of model estimates, non-parametric models are increasingly being utilized, whereby the model structure is also inferred from data. For non-parametric modeling of systems, Gaussian Process (GP) regression \cite{rasmussen2003gaussian}  has been one of the most widely used tools in control theory literature. GP regression keeps track of a Gaussian distribution over infinite dimensional function spaces, in terms of a mean function and a covariance kernel, which are updated with data.  Given any system state, GP regression returns the mean function value at that state, along with a confidence interval. 
% MPC with GP regression modeled uncertainty is presented in \cite{hewing2017cautious,koller2018learning, soloperto2018learning}. 
% However, there is no closed loop constraint satisfaction guarantees in the design, due to only probabilistic knowledge of system uncertainty via GP regression. 
% Recursive system identification and safe exploration for reinforcement learning is presented in \cite{berkenkamp2017safe} under GP regression modeled uncertainty. 
% These deal with choosing safe policies in a region of attraction of state space, which is obtained from the sub-level sets of a known Lyapunov function. These regions are further expanded and new broader safe policies are synthesized, provided the contractivity of the Lyapunov map holds, otherwise the old policy is used as a safe backup.
% We point out that for these algorithms, Lipschitz constant of the system is assumed known (despite the system being unknown), and most importantly, safety is highly probabilistic and often reliant on choice of hyperparameters \cite{deisenroth2013gaussian}, owing to the use of GP regression. 
% GP regression based modeling and recursive system identification is also utilized in \cite{fisac2018general}, for obtaining safe regions in state-space using Hamilton-Jacobi reachability. The method hence not only does not scale well with dimensions, but also is susceptible to inherent failure probabilities of a GP regression based modeling. 
Kernel regression methods such as local linear regression \cite{fan1993local, rosolia2019learning} and Nadaraya-Watson estimator \cite{schuster1979contributions} are some other non-parametric methods for system identification and control. Estimates obtained using these methods often come with confidence intervals as detailed in  \cite{armstrong2018simple}, instead of sets containing all possible realizations of the system, which is a critical drawback from the perspective of robust control.

The focus of this paper is to 
% compute bounds on state dependent system model uncertainty, in presence of measurement noise of finite support. We 
propose a simple non-parametric approach for modelling the \textit{unknown} dynamics of a discrete time autonomous system.
The proposed approach applies to unknown nonlinear systems with dynamics described by a state-update function that is  globally Lipschitz over a bounded domain, with known Lipschitz constant. Instead of identifying the state-update function itself, we identify its graph- the set of all state and corresponding state-update function value pairs. This is done by computing envelopes of the state-update function, which are sets that contain the graph  of the state-update function. 
These envelopes are built by using historical data of state trajectories and the Lipschitz property of the function. 

The paper is divided into two parts. 
In the first part we describe a method to compute the envelope set which contains all possible realizations of the unknown state-update function at any given state. 
% This follows the approach of \cite{berkenkamp2017safe, koller2018learning}.
The authors in \cite{berkenkamp2017safe, koller2018learning} use GP regression modeling to provide probabilistic confidence intervals on the state-update function at any given state. The key difference is that we approximate a function via a subset of the Euclidean space rather than approximating it directly in a function space. 
In the second part, we provide two applications of the proposed approach, namely $(i)$ obtaining tractable set based outer approximations of the unknown state-update function and $(ii)$ computing positive invariant sets \cite{blanchini1999set, kolmanovsky1998theory} for the unknown system using the s-procedure \cite{polik2007survey}. 
%%%%%%%%%%
\section{Notation} \label{sec:notations}
$\Vert\cdot\Vert$ denotes the Euclidean norm in $\mathbb{R}^n$ unless explicitly stated otherwise. An open ball in $\mathbb{R}^n$, of radius $r$ and centered at $x$ is denoted as $\mathcal{B}_r(x)$. Notation $O(\cdot)$ is used to describe an expression that decays to 0 as fast as its argument. The Minkowski sum of two sets $A$ and $B$ is given by 
$$A\oplus B=\{a+b\ | a\in A, b\in B\}.$$We use $\mathrm{ell}(c,R)$ to denote an ellipse that is centered at point $c$ and has a shape matrix $R=R^\top \succ 0$.

%%%%%%%%%%%%%%%%%%%%%%%%%%%%%%%%%%%%%
\section{Problem Formulation}\label{sec:probF}
Consider the discrete time  autonomous, time invariant system 
\begin{align}\label{eq:sys}
    x_{k+1} = f(x_k),
\end{align}
where the state-update function $f(\cdot): \mathcal{X}\rightarrow\mathcal{X}$ describes the system dynamics and is defined over the state space $\mathcal{X} \subseteq \mathbb{R}^n$. 
\begin{assumption}\label{ass:fX_as}
The function $f(\cdot): \mathcal{X}\rightarrow\mathcal{X}$ is continuous and differentiable over a convex and closed domain $\mathcal{X} \subset \mathbb{R}^n$ with $\Vert \nabla f(x)\Vert\leq L$, for all $x\in\mathcal{X}$ and some $L>0$.
\end{assumption}
\begin{proposition}[\cite{rudin1964principles}]\label{prop:lip}
Let Assumption \ref{ass:fX_as} hold. Then $\Vert f(x) - f(y) \Vert \leq L \Vert x-y \Vert$ for all $x,y \in \mathcal{X}$, i.e., $f(\cdot)$ is $L$-Lipschitz in the domain $\mathcal{X}$.
\end{proposition}

Now suppose that the function $f(\cdot)$ is unknown. The objective of this work is to compute a set containing
$f(x)$ for any state $x$ in the state space $\mathcal{X}$ using trajectory data $\{x_0, x_1, x_2,\dots\}$ and the Lipschitz property of the unknown function $f(\cdot)$. 
% Since we are interested in tight bounds, our goal is to improve the bounds as new data $(x_k, f(x_k))$ becomes available.
% As new data $(x_k, f(x_k))$ becomes available, this set is refined.
\begin{assumption}\label{ass:lip}
The Lipschitz constant $L$ is known.
\end{assumption}
In case the Lipschitz constant $L$ is unknown, it can be estimated using methods such as \cite{chakrabarty2019data}. Integrating such estimation methods into the proposed work is a subject of future research.

\begin{remark}\label{rem:mod-unmodsplit}
The problem of characterizing  $L-$Lipschitz un-modelled dynamics $d(\cdot)$ in $$x_{k+1}=\underbrace{\bar{f}(x_k)}_{\textnormal{assumed model}}+\underbrace{d(x_k)}_{\textnormal{un-modelled dynamics}}$$ can also be cast into a problem of the form \eqref{eq:sys}.
In this case, we use the trajectory data $\{x_0,x_1,x_2,\dots\}$ to construct  $\{x_1-\bar{f}(x_0),x_2-\bar{f}(x_1),\dots\}$ which is then used for computing a set containing $d(x)$ at $x\in\mathcal{X}$.
\end{remark}

\section{Proposed Approach}\label{sec:methodology}
We will make use of the following definitions.
\begin{definition}[Graph]\label{def:graph}
The graph of function $f(\cdot): \mathcal{X}\rightarrow \mathcal{X}$ is the set 
\begin{align}\label{graphdef}
    G(f)=\{(x,f(x))\in\mathbb{R}^n\times\mathbb{R}^n\vert\ \forall x\in\mathcal{X}\}.
\end{align}
\end{definition}

\begin{definition}[Envelope]\label{def:envelope}
An envelope of function $f(\cdot):\mathcal{X}\rightarrow\mathcal{X}$ is a set $\mathbf{E}^f\subseteq \mathbb{R}^n\times\mathbb{R}^n$, with the property
\begin{align}
    G(f)\subseteq \mathbf{E}^f.
\end{align}
\end{definition}

%Instead of learning a function on the %infinite dimensional space of %$L-$Lipschitz functions %$f(\cdot):\mathcal{X}\rightarrow\mathcal%{X}$, we obtain set-valued maps from %$\mathcal{X}$ to $2^\mathcal{X}$. These %maps are in turn, derived from subsets %of $\mathbb{R}^n\times\mathbb{R}^n$ (the %envelopes) that contain the graph of %$F(\cdot)$, and are refined iteratively %based on recorded measurements. \\

We use trajectory data $\{x_0,x_1,x_2,\dots\}$ of the system dynamics \eqref{eq:sys} to construct an envelope of the system dynamics $f(\cdot)$. Observe that the trajectory data can be used to construct tuples $(x_k,f(x_k))=(x_k,x_{k+1})$.
In particular, at every time instant $N$, we have access to measurements $(x_{k},f(x_{k}))$, for all $k=0,1,\dots, N-1$. These measurements are utilized to construct envelopes recursively. Our approach for envelope construction is summarised as follows:
\begin{enumerate}
    \item At time $N$, compute an envelope $\mathcal{E}(x_{N-1})$ using the tuple $(x_{N-1},f(x_{N-1}))$ and the $L-$Lipschitz property of $f(\cdot)$.
    
    \item  Compute a refined envelope $\mathbf{E}^f_{N}$ by intersecting the envelope $\mathbf{E}^f_{N-1}$ from  time $N-1$ with the envelope $\mathcal{E}(x_{N-1})$ computed in step 1), i.e.,
    $$\mathbf{E}^f_N=\mathbf{E}^f_{N-1}\cap\mathcal{E}(x_{N-1}).$$
\end{enumerate}
For 1), the envelope is obtained as the sublevel set of a quadratic function. Afterwards, 2) is obtained by using the set membership approach \cite{tanaskovic2014adaptive, bujarbaruahAdapFIR, bujarArxivAdap, bertsekas1971recursive}. Finally, we use the computed envelope for obtaining a set containing the value of $f(x)$ at any $x\in\mathcal{X}$, using the notion of a \emph{slice} of an envelope defined below.
\begin{definition}[Envelope Slice]\label{def:env_slice}
The slice of an envelope $\mathbf{E}^f\subseteq\mathbb{R}^n\times\mathbb{R}^n$ at a given $\bar{x}\in\mathcal{X}$ is the set defined as
\begin{align}
\mathbf{E}^f\Big\vert_{x=\bar{x}}=\{(x,y)\in\mathbf{E}^f\subseteq\mathbb{R}^n\times\mathbb{R}^n\vert x=\bar{x}\}.
\end{align}

\end{definition}

Fig.~\ref{fig:graph} shows a typical realization of the proposed approach along with the associated set definitions which are detailed next.

\begin{figure}[h]
    \centering
    \includegraphics[width=\columnwidth]{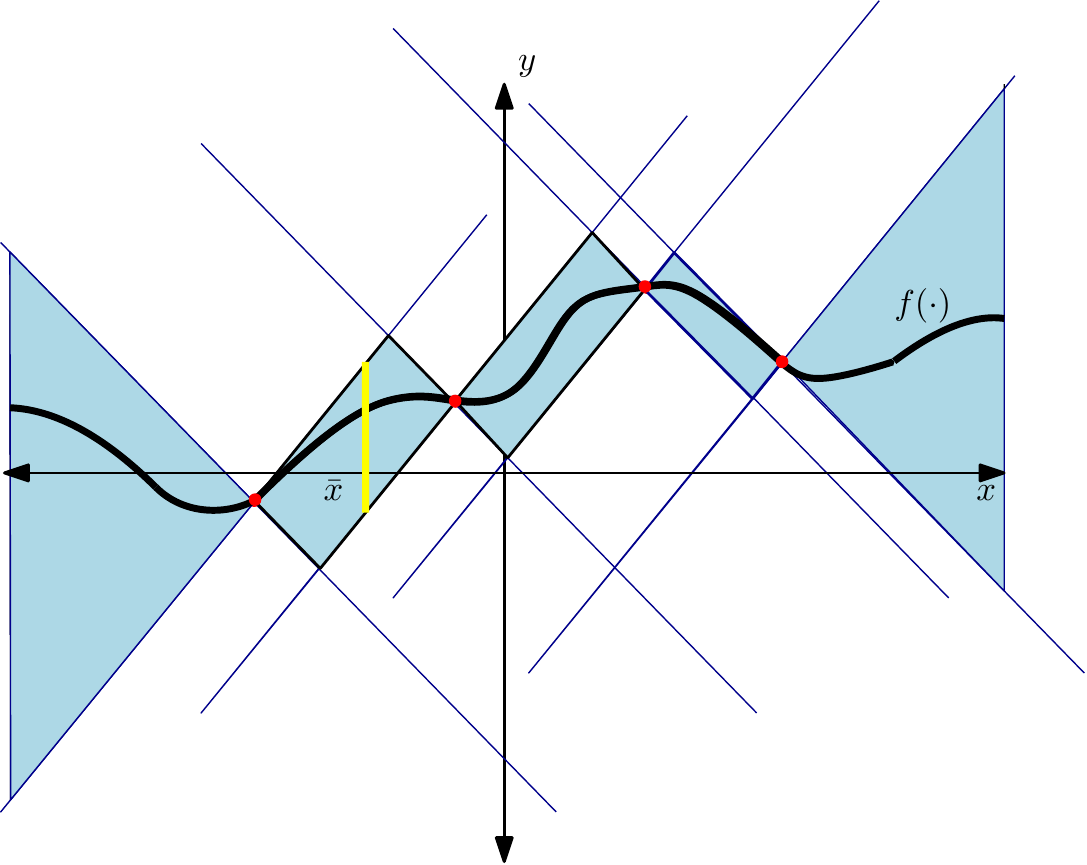}
    \caption{Construction of an envelope for a one dimensional system to approximate the graph $G(f)$ (black curve) of its state-update function $f(\cdot)$. Tuples $(x,f(x))$ (red points) obtained from trajectory data are used for constructing the envelope (blue set) and its slice (yellow set) at $x=\bar{x}$.}
    \label{fig:graph}
\end{figure}
\subsection{Envelope Construction}
Inspired by \cite{fazlyab2019safety,megretski1997system}, we use quadratic constraints (QCs) as our main tool to approximate the graph of a function. A definition appropriate for our purposes is presented below. 
\begin{definition}[QC Satisfaction] A set $\mathbb{X}\subset \mathbb{R}^{n}$  satisfies the quadratic constraint specified by a symmetric matrix $Q$ if
\begin{equation}
    \begin{bmatrix}x\\ 1\end{bmatrix}^\top Q\begin{bmatrix}x\\ 1\end{bmatrix}\leq 0,\quad \forall x\in\mathbb{X}.
\end{equation}
\end{definition}
\vspace{2mm}

The following proposition uses a QC to characterize a coarse approximation of the graph of an $L-$Lipschitz function.
\begin{proposition}\label{prop:qc}
The graph $G(f)$ of an $L-$Lipschitz function $f(\cdot)$ satisfies the QC specified by the matrix
\small
\begin{align}\label{qc}
Q_L^f(x_k)=\begin{bmatrix}
-L^2\mathbb{I}_n& \mathbf{0}_{n\times n} & L^2x_k\\
\mathbf{0}_{n\times n} & \mathbb{I}_n & -f(x_k)\\
L^2 x_k^\top & -f^\top(x_k) & -L^2  x_k^\top x_k\\
& & \hspace{10mm} +f^\top(x_k) f(x_k)
\end{bmatrix},
\end{align}
\normalsize
for any $(x_k, f(x_k))\in G(f)$.\\
\end{proposition}

\begin{proof}
Using the definition of the $L-$Lipschitz property of $f(\cdot)$ (from Proposition \ref{prop:lip}), at any $(x_k,f(x_k)) \in G(f)$, we have
\small
\begin{align*}
    &\Vert f(x)-f(x_k)\Vert \leq L \Vert x-x_k \Vert \quad ~~~\forall (x,f(x))\in G(f),\\
    \iff & (f(x)-f(x_k))^\top(f(x)-f(x_k))\leq L^2(x-x_k)^\top(x-x_k)\\
    &~~~~~~~~~~~~~~~~~~~~~~~~~~~~~~~~~~~~~~~~~~\forall (x,f(x))\in G(f),\\
    \iff &\begin{bmatrix}x\\f(x)\\1 \end{bmatrix}^\top Q_L^f(x_k)\begin{bmatrix}x\\f(x)\\1 \end{bmatrix}\leq 0,\quad\forall (x,f(x))\in G(f).
\end{align*}
\normalsize
Therefore $G(f)$ satisfies the QC specified by $Q_L^f(x_k)$.
\end{proof}
The following corollary then gives us the definition of the envelope $\mathcal{E}(x_k)$.
\begin{corollary}\label{corr:env}
The set defined by \begin{align}\label{eq:gapprox}
\mathcal{E}(x_k)=\{(x,y)\in\mathbb{R}^{n}\times\mathbb{R}^{n}\vert \begin{bmatrix}x\\y\\1 \end{bmatrix}^\top Q_L^f(x_k)\begin{bmatrix}x\\y\\1 \end{bmatrix}\leq 0\}\end{align}
is an envelope for all $L-$Lipschitz functions that pass through $(x_k, f(x_k))$.
\end{corollary}
\begin{proof}
Let $g(\cdot)$ be any $L-$ Lipschitz function such that $g(x_k)=f(x_k)$. From the definition of Lipschitz property we have 
\small
\begin{align*}
    &\Vert g(x)-f(x_k)\Vert \leq L \Vert x-x_k \Vert, \quad \forall (x,g(x))\in G(g),\\
    \iff &\begin{bmatrix}x\\g(x)\\1 \end{bmatrix}^\top Q_L^f(x_k)\begin{bmatrix}x\\g(x)\\1 \end{bmatrix}\leq 0,\quad\forall (x,g(x))\in G(g),\\
    \iff & (x,g(x))\in\mathcal{E}(x_k),\quad\forall (x,g(x))\in G(g),\\
    \iff & G(g)\subseteq \mathcal{E}(x_k).
    \normalsize
\end{align*}
\end{proof}

\begin{remark} The proposed formulation can also be extended to accommodate bounded noise in the measurements of $x_k$ in \eqref{eq:sys}. Suppose that the measurement model is given by $$z_k=x_k+w_k,$$ where $w_k$ belongs to a compact set $\mathcal{W}$. Then the envelope that is guaranteed to contain $G(f)$ is given by $\mathcal{E}(z_k)\oplus (\mathcal{W}\times\mathcal{W})$ where $Q_L^f(\cdot)$ is now constructed using $(z_k, z_{k+1})$.
\end{remark}

\subsection{Successive Graph Approximation}
At time $N$, the envelope $\mathcal{E}(x_{N-1}))$ constructed in \eqref{eq:gapprox} using the tuple $(x_{N-1},f(x_{N-1}))$ can now be used to recursively compute a new envelope $\mathbf{E}^f_N$ by refining the envelope $\mathbf{E}^f_{N-1}$ from time $N-1$ via set intersection-
\begin{align}\label{eq:rer}
\mathbf{E}^f_N=\mathbf{E}^f_{N-1}\cap\mathcal{E}(x_{N-1})
\end{align}
In the following lemma we show that the sets computed in this fashion are indeed envelopes.
\begin{lemma}\label{lem:inter}
For $N\geq1$, given a sequence $\{x_k\}_{k=0}^{N-1}$ obtained under the dynamics (\ref{eq:sys}), we have 
\begin{align}\label{eq:intenv}
    G(f)\subseteq\mathbf{E}^f_N=\mathbf{E}^f_{N-1}\cap\mathcal{E}(x_{N-1})=\bigcap_{k=0}^{N-1} \mathcal{E}(x_k).
\end{align}
\end{lemma}
% \textit{Proof deferred to appendix}\\
\begin{proof}
See Appendix.
\end{proof}
The recursion is initialized with the trivial envelope $\mathbf{E}^f_0=\mathbb{R}^n\times\mathbb{R}^n$. The procedure is described in Algorithm~\ref{alg:RER}. 

\begin{algorithm}[h]
\DontPrintSemicolon

\KwInit{ $\mathbf{E}_0^f=\mathbb{R}^{n}\times\mathbb{R}^n$}

\KwInput{ At time $N$, $(x_{N-1}, f(x_{N-1}))$ and $\mathbf{E}^f_{N-1}$ }
\KwOutput{Approximation of $G(f)$ at time $N$: $\mathbf{E}^f_N$}

Compute $Q_L^f(x_{N-1})$ (from \eqref{qc}) \;
Compute $\mathcal{E}(x_{N-1})$ using $Q_L^f(x_{N-1})$ (from \eqref{eq:gapprox}) \;
Set $\mathbf{E}^f_{N}=\mathbf{E}^f_{N-1}\cap\mathcal{E}(x_{N-1})$
 \caption{Recursive Envelope Refinement}
 \label{alg:RER}
\end{algorithm}

Note that since the envelope at any time $N$ is computed by intersecting with the envelope at time $N-1$, they are getting successively refined, i.e.,
\begin{align}\label{eq:envelorder}
\mathbf{E}^f_N\subseteq \mathbf{E}^f_{N-1}\subseteq\mathbf{E}^f_{N-2}\dots\subseteq\mathbf{E}^f_0
\end{align}
Now we provide a condition under which the shrinking sets generated by recursion \eqref{eq:rer} stop shrinking i.e., recursion \eqref{eq:rer} attains a fixed point. Intuitively, we would expect this to happen when the incoming tuples $(x,f(x))$ constructed from trajectory data have already been seen previously. The following definition formalises the notion of such trajectories.
\begin{definition}[Periodic Orbit \cite{zhou2003periodic}]\label{def:porb} A $p$-periodic orbit of the discrete dynamical system \eqref{eq:sys} is the set of states obtained under dynamics $x_{k+1}=f(x_k)$ with the property that $x_k=x_{k+p}$ for some finite $p\geq1$ and for all $k\geq0$, i.e.,
\begin{align}\label{eq:porb}
    \mathcal{O}_p(x_0)&=\{x\in\mathcal{X}\ \vert\ x_{k+1}=f(x_k), x_k=x_{k+p}, \nonumber\\
    &~~~~~~~~~~~~~~~~~~~~~~~~~~~~x=x_k,\forall k\geq0\}. 
\end{align}

\end{definition}
Note that the set $\mathcal{O}_p(\bar{x}_\textnormal{eq})=\{\bar{x}_\textnormal{eq}\}$ for all $p\geq 1$ where $\bar{x}_\textnormal{eq}$ is the fixed point $\bar{x}_\textnormal{eq}=f(\bar{x}_\textnormal{eq})$ of system \eqref{eq:sys}. Associated to each fixed point, one can define the set of states that converge to it as follows. 
\begin{definition}[Domain of Attraction\ \cite{ortega1973stability}]\label{def:doa}
The domain of attraction of fixed point $\bar{x}_\textnormal{eq}$ is defined as the set $$\mathbf{D}(\bar{x}_\textnormal{eq})=\{x\in\mathcal{X}\vert x_{k+1}=f(x_k),\lim_{k\rightarrow\infty}x_k=\bar{x}_\textnormal{eq},x=x_0 \}.$$
\end{definition}
The following proposition uses Definition~\ref{def:porb} and Definition~\ref{def:doa} to identify sufficient conditions on system trajectories for termination of the recursion \eqref{eq:rer}. 
\begin{proposition}\label{lem:rerfp}
Given a system trajectory denoted by the set $\{x_0,x_1,x_2,\dots\}$, the recursion \eqref{eq:rer} has a fixed point if either of the following conditions hold: 
\begin{enumerate}
    \item $\mathcal{O}_p(x_k)\subseteq\{x_0,x_1,x_2,\dots\} $ for some finite $p\geq1$ and some $k\geq 0$.
    \item $x_0\in\mathbf{D}(\bar{x}_\textnormal{eq})$ for some fixed point $\bar{x}_\textnormal{eq}$.
\end{enumerate}
\end{proposition}
\begin{proof}
See Appendix.
\end{proof}
Next we present how the envelope slice is derived from the constructed envelopes for obtaining a set-valued estimate of $f(x)$ at any $x\in\mathcal{X}$

%%%%%%%%%%%%%%%%%%%%%%%%%%%%%%%%%
\subsection{Envelope Slice Computation}
\label{ssec:nom_function}
% Our goal is to estimate the unknown function $F(\cdot)$, given access to measurements $F(x^s)$ at any sampled point $x^s$. We 
% qualify the function estimation with 
% two key aspects: $(i)$ estimating a nominal function $\bar{F}(\cdot)$, and $(ii)$ 
% obtain a set where the possible value of $F(\bar{x})$ can lie, given any queried $\bar{x} \in \mathcal{X}$. 
% \subsubsection*{Robust Set Estimation} 
The set of possible values of function $f(\bar{x})$ at any $\bar{x}\in\mathcal{X}$ can be obtained using \eqref{eq:gapprox} from the function values $f(x_k)$ collected at $k = \{0,1,2,\dots,N-1\}$.  From \emph{only} the $k$-th measurement, we can obtain an estimate of the set of possible values of $f(\bar{x})$, by constructing the slice of envelope $\mathcal{E}(x_k)$ at $x=\bar{x}$, from Definition~\ref{def:env_slice}. We denote this slice with the set $\mathcal{S}(x_k, \bar{x})$ as
\begin{align}\label{eq:sampl_range}
    \mathcal{S}(x_k, \bar{x}) & = \mathcal{E}(x_k) \Big \vert_{x = \bar{x}},\nonumber\\
    & =\{y\in\mathcal{X}\vert  \begin{bmatrix} \bar{x} \\ y \\ 1 \end{bmatrix}^\top Q^f_L(x_k) \begin{bmatrix} \bar{x} \\ y \\ 1 \end{bmatrix}  \leq 0\}\nonumber\\
 &=\{y\in\mathcal{X}\vert \begin{bmatrix}
    y \\ 1 \end{bmatrix}^\top \bar{A}(k,\bar{x}) \begin{bmatrix}
    y \\ 1 \end{bmatrix} \leq 0\},
\end{align}
where we have denoted $\bar{A}(k,\bar{x}) =  M Q^f_L(x_k)M^\top -L^2\begin{bmatrix}0&0\\0&(\bar{x}-2x_k)^\top \bar{x} \end{bmatrix}$, for any $k = \{0,1,2,\dots,N-1 \}$, with matrix $M = \begin{bmatrix}0&1&0\\0&0&1\end{bmatrix}$. Corollary~\ref{corr:env} then implies $f(x) \in \mathcal{S}(x_k, x)$ at any $x\in\mathcal{X}$. 
\begin{proposition}\label{prop:ball_range}
At any $\bar{x}\in\mathcal{X}$, $\mathcal{S}(x_k, \bar{x})$ is a closed norm ball of radius $L \Vert \bar{x} - x_k \Vert$, centered at $f(x_k)$ for each $k=\{0,1,\dots,N-1\}$.
\end{proposition}
\begin{proof}
Expanding out \eqref{eq:sampl_range} gives us
\begin{align}\label{eq:fball}
    \Vert f(\bar{x}) - f(x_k) \Vert^2 \leq L^2 \Vert \bar{x} - x_k \Vert^2,
\end{align}
for each $k=\{0,1,\dots,N-1\}$ and thus proves the claim.
\end{proof}
As we successively collect data points $(x_k,f(x_k))$ for $k=\{0,1,2,\dots, N-1 \}$ under dynamics \eqref{eq:sys}, the set of possible values of $f(\bar{x})$ at any $\bar{x}\in\mathcal{X}$ is refined as
\begin{align}\label{eq:func_dom}
    \mathcal{F}_N(\bar{x}) = \bigcap_{k=0}^{N-1} \mathcal{S}(x_k, \bar{x})=\bigcap_{k=0}^{N-1} \mathcal{E}(x_k) \Big \vert_{x = \bar{x}}=\mathbf{E}_N^f\Big \vert_{x = \bar{x}},
\end{align}
with the guarantee $f(\bar{x}) \in \mathcal{F}_N(\bar{x})$ at any given time $N \geq 1$. Notice that $\mathcal{F}_N(\bar{x})$ is a slice of envelope $\mathbf{E}^f_N$ at $x=\bar{x}$, as per Definition~\ref{def:env_slice}. We further note that $\mathcal{F}_N(\bar{x})$ in \eqref{eq:func_dom} is convex and compact, as it is an intersection of convex and compact sets \eqref{eq:fball}.
\par
So far we have seen that the envelopes generated by Algorithm \ref{alg:RER} are getting successively refined (in \eqref{eq:envelorder}) and possibly stop improving (as noted in proposition~\ref{lem:rerfp}). But given a trajectory that yields a terminating recursion \eqref{eq:rer}, where in the state space $\mathcal{X}$ are the envelope slices ``tight"? We use the notion of the diameter of a compact set (\cite{rudin1964principles}) to quantify ``tightness" or size of the envelope slice. In the following theorem, we show that if a trajectory starts in the domain of attraction $\mathbf{D}(\bar{x}_\textnormal{eq})$ of a fixed point $\bar{x}_\textnormal{eq}$ of \eqref{eq:sys}, then the error in approximation of $G(f)$ by $\mathbf{E}^f_N$ at points arbitrarily close to $\bar{x}_\textnormal{eq}$ (measured by the diameter of the envelope slice $\mathcal{F}_N(x)$ at any $x\in\mathcal{X}$), gets arbitrarily small for large enough $N$.
\begin{theorem}\label{thrm:AsympApprox}
Suppose we are given a system trajectory denoted by the set $\{x_0,x_1,x_2,\dots,x_N\}$ where $x_0\in\mathbf{D}(\bar{x}_\textnormal{eq})$.
Then for states $x$ arbitrarily close to $\bar{x}_\textnormal{eq}$,  the diameter of $\mathcal{F}_N(x)$ is arbitrarily small for large enough $N$, i.e.,
\begin{align*}
  &\forall\epsilon>0, \exists~\bar{N}(\epsilon):\\  
  &\max_{y\in\mathcal{F}_N(x)}\Vert  y - f(x)  \Vert = O(\epsilon), \\
  &\forall  x\in\mathcal{B}_{\epsilon}(\bar{x}_\textnormal{eq}),\ \forall N \geq \bar{N}(\epsilon).
\end{align*}

\end{theorem}
\begin{proof}
See Appendix.
\end{proof}

% \subsubsection*{Nominal Function Estimation}
% A nominal function $\bar{F}_k(\cdot)$ is estimated along with the set $\mathcal{F}_k(\cdot)$ at any $k \geq 0$. After \eqref{eq:func_dom} is obtained at any query $\bar{x}$, $\bar{F}_k(\bar{x})$ is computed as the Chebyshev center of this set. While the problem of obtaining the Chebyshev center is a convex optimization problem, it may not be suitable for an online implementation of the proposed framework. Towards this end, we make use of tools from ellipsoidal calculus \cite{kurzhanskiui1997ellipsoidal} to obtain our function estimate.   

\section{Applications} \label{sec:app}
In this section we demonstrate two applications and corresponding computationally tractable algorithms that utilize the proposed approach in the paper.

\subsection{Ellipsoidal Outer Approximation of $\mathcal{F}_N(\bar{x})$}\label{ssec:rob_opt}
In order to design computationally  tractable robust optimization \cite{ben2009robust} algorithms for all realizations of $f(x)$ at any $x\in\mathcal{X}$ and $N \geq 1$, one must have a ``nice" geometric representation of the envelope slice $\mathcal{F}_N(\bar{x}) = \mathbf{E}^f_N\Big \vert_{x = \bar{x}}$, for all $N \geq 1$. We hereby propose an approach to obtain an ellipsoidal outer approximation to $\mathcal{F}_N(\bar{x})$ for any $N \geq 1$ using the s-procedure  \cite[Section~11.4]{calafiore2014optimization}, having collected measurements at $k=0,1,\dots, N-1$,

Let us parametrize an \emph{ellipsoidal} outer approximation of $\mathcal{F}_N(x_N)$, which we denote by \textnormal{ell}($c(\bar{x}), R(\bar{x})$) as
\begin{align*}
 &   \textnormal{ell}(c(\bar{x}), R(\bar{x})) = \\ & \{y\in\mathbb{R}^n\vert\ (y-c(\bar{x}))^\top R^{-1}(\bar{x})(y-c(\bar{x})) \leq 1\},
\end{align*}
where vector $c(\bar{x})$ and matrix $R(\bar{x})$ are the decision variables, and are functions of $\bar{x}$. We seek the smallest ellipsoidal set such that
\begin{align*}
   \mathcal{F}_N(\bar{x}) =  \bigcap_{k=0}^{N-1}\mathcal{S}_k(x_k, \bar{x})\subseteq  \textnormal{ell}(c(\bar{x}), R(\bar{x})).
\end{align*}
 From s-procedure \cite{Boyd:2004:CO:993483} we know that the above holds true, if there exists scalars $\{ \tau_0, \tau_1, \dots, \tau_{N-1} \} \geq 0$ such that
\begin{align}\label{eq:s_proc}
&\begin{bmatrix} R^{-1}(\bar{x}) & -R^{-1}(\bar{x}) c(\bar{x}) \\ -c^\top(\bar{x})R^{-1}(\bar{x}) & c^\top(\bar{x})R^{-1}(\bar{x})c(\bar{x})\\
& \hspace{30mm} -1\end{bmatrix}\nonumber\\ &~~~~~~~~~~~~~~~~~~~~~~~~~~~~~~~~~~~~~-\sum \limits_{k = 0}^{N-1} \tau_k \bar{A}^s(k,{\bar{x}}) \preceq 0.
\end{align}
We reformulate the above feasibility problem \eqref{eq:s_proc} as a semi-definite program (SDP) in the appendix.

\subsection{Positive Invariant Set Computation}\label{ssec:app2}
\begin{figure}[h]
    \centering
    \includegraphics[width=0.7\columnwidth]{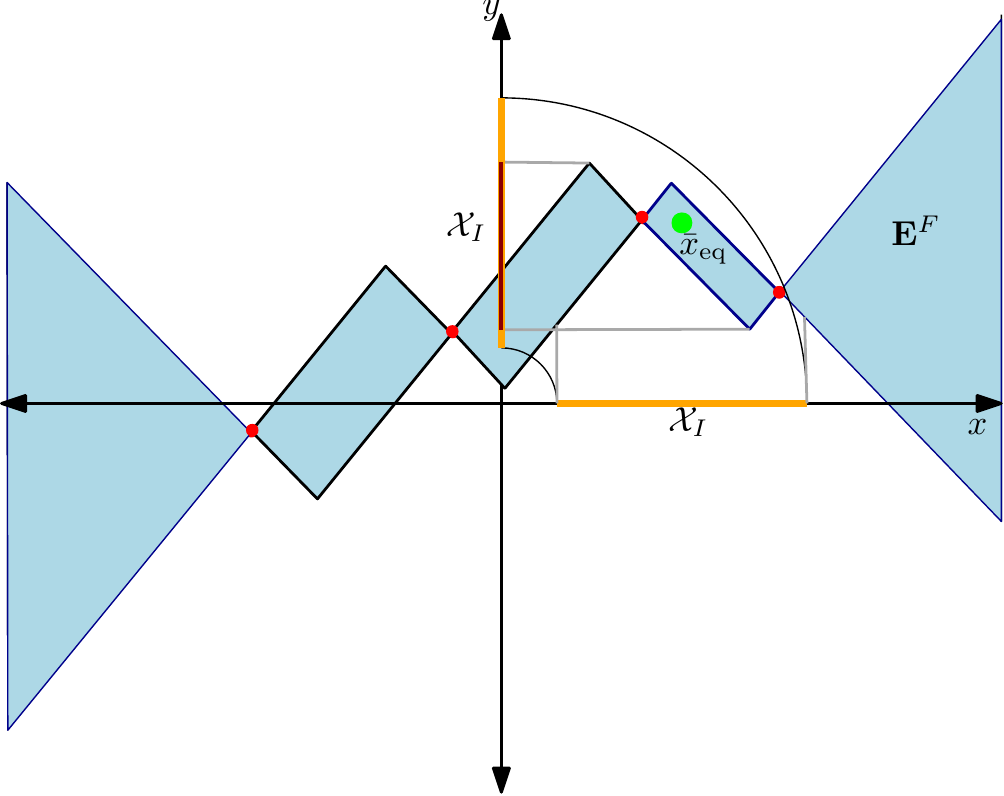}
    \caption{Illustration of invariance using an envelope of system dynamics. $\mathcal{X}_I$(orange) is invariant for all $f(\cdot)$ with $G(f)\subset\mathbf{E}^f$(blue). }
    \label{fig:pos_inv}
\end{figure}
\begin{definition}
[Positive Invariant Set] A set  $\mathcal{X}_I \subseteq \mathcal{X}$ is said to be positive invariant for the system with dynamics \eqref{eq:sys} if
\begin{align*}
 x\in\mathcal{X}_I\Rightarrow f(x) \in \mathcal{X}_I,
\end{align*}
i.e. the set $\mathcal{X}_I$ maps to itself, under the dynamics map $f(\cdot):\mathcal{X}\rightarrow\mathcal{X}$.
\end{definition}

Let there be an equilibrium point $\bar{x}_\textnormal{eq}$ defined as $f(\bar{x}_\textnormal{eq}) = \bar{x}_\textnormal{eq}$. We wish to characterize a positive invariant set $\mathcal{X}_I \subseteq \mathcal{X}$ containing this equilibrium. For the sake of computational tractability, we represent this set as an intersection of $n_I$ ellipsoids, centered at $\bar{x}_\textnormal{eq}$. 
% The property we want in this set is: $\{ \bar{x} \in \mathcal{X}_I: F(\bar{x}) \subseteq \mathcal{X}_I$\}, implying that the set $\mathcal{X}_I$ maps to itself or a subset under the map of the function $F(\cdot)$, thus making it invariant under $F(\cdot)$. 
We parameterize the invariant set for some $P_j\succ 0$ with $j=1,2,\dots n_I$ as follows
\begin{align}\label{eq:invar_par}
    \mathcal{X}_I=\{x \in\mathbb{R}^n\vert&\ \begin{bmatrix}x\\1\end{bmatrix}^\top\begin{bmatrix}P_j&-P_j\bar{x}_\textnormal{eq}\\-\bar{x}^{\top}_\textnormal{eq}P_j & \bar{x}^\top_\textnormal{eq}P_j\bar{x}_\textnormal{eq}-1\end{bmatrix}\begin{bmatrix}x\\1\end{bmatrix}\leq 0, \nonumber\ \\
    &~~~~~~~~~~~~~~~~~~~~~~~~~~~~~\forall j=1,2,\dots n_I\}. 
\end{align}
% We aim to construct the largest subset $\mathcal{X}_I$ of the state space $\mathcal{X}$ parameterized as in \eqref{eq:invar_par} such that it maps to itself under the dynamics map $F:\mathcal{X}\rightarrow\mathcal{X}$. 
Observing that the tuple $(x, f(x))\in G(f)\subseteq \mathbf{E}^f_N $ consists of a point $x\in\mathcal{X}$ and its image $f(x)\in\mathcal{X}$ under the map $f(\cdot)$, we can use the collected $Q_L^f(x_k)$'s for $k=\{0,1,\dots,N-1\}$ at any time $N$, to obtain a sufficiency condition for \eqref{eq:invar_par} as detailed in the following proposition. 
%%%%%%%%%%%%%%%%%%%%%%%%%%%%%%%%%%%%%%%%%%%%

\begin{proposition}\label{prop:invsetLMI}
Suppose that we are given an approximation of $G(f)$ at time $N$, i.e., $\mathbf{E}^f_N=\bigcap_{k=0}^{N-1} \mathcal{E}(x_k)$, constructed by Algorithm~\ref{alg:RER}. If there exists $\tau_k\geq0$, for all $k=0,1,\dots N-1$ and $P_j\succ 0$, $\rho_{jm}>0$ for all $j,m=1,2,\dots, n_I$, such that the following Bilinear Matrix Inequality (BMI) is feasible
\begin{align}\label{eq:s_proc_invar}
\sum_{j=1}^{n_I}&\begin{bmatrix} -\rho_{jm}P_j& 0 & \rho_{jm}P_j\bar{x}_\textnormal{eq}\\0& P_m & -P_m\bar{x}_\textnormal{eq}\\\bar{x}^\top_\textnormal{eq}\rho_{jm}P_j&-\bar{x}^\top_\textnormal{eq}P_m&-\bar{x}^\top_\textnormal{eq}(\rho_{jm}P_j-P_m)\bar{x}_\textnormal{eq}\\
& & \hspace{10mm} +(\rho_{jm}-1) \end{bmatrix}\nonumber\\ & ~~~~~~~~~~~~~~~~~~~~~~~~~~~~~~~~- \sum \limits_{k=0}^{N-1} \tau_k Q_L^f(x_k) \preceq 0,
\end{align}
for all $m=1,2,\dots, n_I$, then the set $\mathcal{X}_I$ is a positive invariant set for the system with dynamics \eqref{eq:sys}.
\end{proposition}
\begin{proof}
See Appendix.
\end{proof}
One of many approaches to solving such a BMI (see \cite{andySOS}) is detailed in the 
appendix.

\section{Numerical Example}
In this section we demonstrate the approach proposed in Section \ref{sec:methodology} for characterizing the un-modelled dynamics of a pendulum. We also showcase construction of a positive invariant set for this system, utilizing the tools from Section~\ref{ssec:app2}.

\subsection{Pendulum Model}
The continuous time model of the considered pendulum is given by
\begin{align}\label{eq:pend}
    ml^2 \ddot{\theta} + mgl\sin{\theta} + \tilde{d}(\theta, \dot{\theta}) = \mathcal{T},
\end{align}
where $m$ is the mass, $l$ is the length, $\theta$ is the angle the pendulum makes with the vertical axis, $\tilde{d}(\theta,\dot{\theta})$ is an un-modelled damping force with known Lipschitz constant $L_d$ and $\mathcal{T}$ is a known external torque. In this work, we simulate the system with the damping force $\tilde{d}(\theta, \dot{\theta})=-L_d\dot{\theta}$ and characterize state-dependent bounds for the same. We write the pendulum dynamics \eqref{eq:pend} in state-space form as
\begin{align}\label{eq:pend_ss}
    \begin{bmatrix} \dot{\theta}\\ \ddot{\theta}\end{bmatrix} = \begin{bmatrix} 0 & I \\ 0 & 0 \end{bmatrix} \begin{bmatrix} {\theta} \\  \dot{\theta}\end{bmatrix} + \begin{bmatrix} 0 \\ \frac{\mathcal{T}}{ml^2}-\frac{g}{l}\sin \theta - \frac{\tilde{d}(\theta,\dot{\theta})}{ml^2} \end{bmatrix},
\end{align}
where $x = [\theta~ \dot{\theta}]^\top$ is the state of the pendulum. We consider a torque $\mathcal{T}$ that stabilizes the pendulum's state when it's upright, i.e., when $\bar{x}_{\mathrm{eq}}=[\pi~ 0]^\top$. We discretize system \eqref{eq:pend_ss} and write it in the form of \eqref{eq:sys} as $x_{k+1} = f(x_k)$.
% , where $F(\cdot)$ contains the unknown, un-modelled dynamics $\tilde{d}(x)$.
% \begin{table}[h]
%  	\renewcommand{\arraystretch}{2.2}
% 	\caption{Simulation Parameters}
% 	\label{table:sim_par}
% 	\centering
%  \begin{tabular}{|c | c|} 
%  \hline
%  Parameter & Value  \\ [0.5ex]
%  \hline\hline
%  $m$ & $2\ \mathrm{kg}$ \\ 
%  \hline
%  $l$ & $2\ \mathrm{m}$ \\
%  \hline
%   $g$ & $9.8\ \mathrm{m/s^{2}}$ \\
% %  \hline
% %  $dt$ & $0.05\ \mathrm{s}$ \\
%  \hline
%   $L_d$ & $0.1\ \mathrm{N}$ \\
%  \hline
% \end{tabular}
% \end{table}
We then simulate the system forward in time with a variational integrator for mechanical systems, as in \cite{nair2019discrete}. The simulation parameters are: $m = 2\mathrm{kg}, l = 2\mathrm{m}$ and $L_d = 0.2\mathrm{N}$.

\subsection{Envelope Construction for Damping Force}\label{sysID_Pend}
The discrete time model $x_{k+1}=f(x_k)$ is decomposed as 
\begin{align}
    x_{k+1}=\underbrace{\bar{f}(x_k)}_{\textnormal{assumed model}}+\underbrace{d(x_k)}_{\textnormal{un-modelled damping}},
\end{align}
where $x_k = [\theta_k~\dot{\theta}_k]^\top$ , $d(\cdot)$ is the unknown damping in discrete time with Lipschitz constant $\tilde{L}_d=\frac{L_dT_S}{ml^2}$ and $T_S=0.005 s$ is the sampling period. Our experiment is succinctly described below:
\begin{itemize}
\item Trajectories up to a specified time instant $N$, starting from four different initial conditions \textcolor{black}{$x_0 = \{[\frac{5 \pi}{6}~ 0]^\top,[\frac{5 \pi}{3}~  -0.5]^\top,[\frac{\pi}{6}~ 0]^\top,[\frac{5 \pi}{4}~-0.2]^\top\}$} are simulated (solid lines in Fig. \ref{fig:traj_data}) and stored. 
\item Realizations of the un-modelled dynamics $d(x_k)$ are recorded via the measurement model $d(x_k)=x_{k+1}-\bar{F}(x_k)$. 
% The horizontal plane depicts the $[\theta,\dot{\theta}]$ space and the four trajectories recorded. 
\item Having recorded the measurements $(x_k,d(x_k))$ for $k=0,1,\dots,N-1$ along all four trajectories, we construct the estimate $\mathcal{D}_N(\bar{x})$ (defined as in \eqref{eq:func_dom}) of $d(\bar{x})$  at six different query points ($\star$ and $\circ$ in Fig. \ref{fig:traj_data}) using Algorithm~\ref{alg:RER}.
\end{itemize}
\begin{table}[h]
 	\renewcommand{\arraystretch}{2.7}
	\caption{Uncertainty range (up to three significant digits). Symbol $[\cdot, \cdot]$ denotes an interval}
	\centering
 \begin{tabular}{|c|c|c|} 
 \hline
 Query Point $\bar{x}$ & $\mathcal{D}_{100}(\bar{x}) / 10^{-4}$ & $\mathcal{D}_{4000}(\bar{x}) / 10^{-4}$  \\ [0.5ex]
 \hline\hline
 $[2.12\ -0.45]^\top$ & $[-0.837,0.759]$ &  $[-0.001,0.759]$\\ 
 \hline
 $[3.11\ 0.84]^\top$ & $[-1.22,0.73]$ &  $[-1.04,-1.04]$\\ 
 \hline
  $[1.40\ 0.34]^\top$ & $[-1.58,0.79]$  & $[-0.43,-0.43]$\\
%  \hline
%  $dt$ & $0.05\ \mathrm{s}$ \\
 \hline
  $[3.05\ -0.37]^\top$ & $[-0.708,0.486]$ &  $[0.46,0.46]$\\ 
 \hline
 $[4.21\ 0.38]^\top$ & $[-2.05,0.74]$ &  $[-0.56,0.16]$\\ 
 \hline
 $[5.60\ 0.22]^\top$ & $[-3.73,2.46]$ &  $[-0.28,-0.28]$\\
 \hline
\end{tabular}
	\label{table:uncert_estim}
\end{table}
\textcolor{black}{From Table~\ref{table:uncert_estim}, we observe that the range of un-modelled dynamics $\mathcal{D}_N(\bar{x})$ shrinks at all query points $\bar{x}$, as more data is collected. This is a direct consequence of the fact that as shown in \eqref{eq:func_dom}, $\mathcal{D}_N(\bar{x})$ is obtained with successive intersection operations upon gathering new measurements. 
\begin{figure}[h]
    \includegraphics[width=\columnwidth]{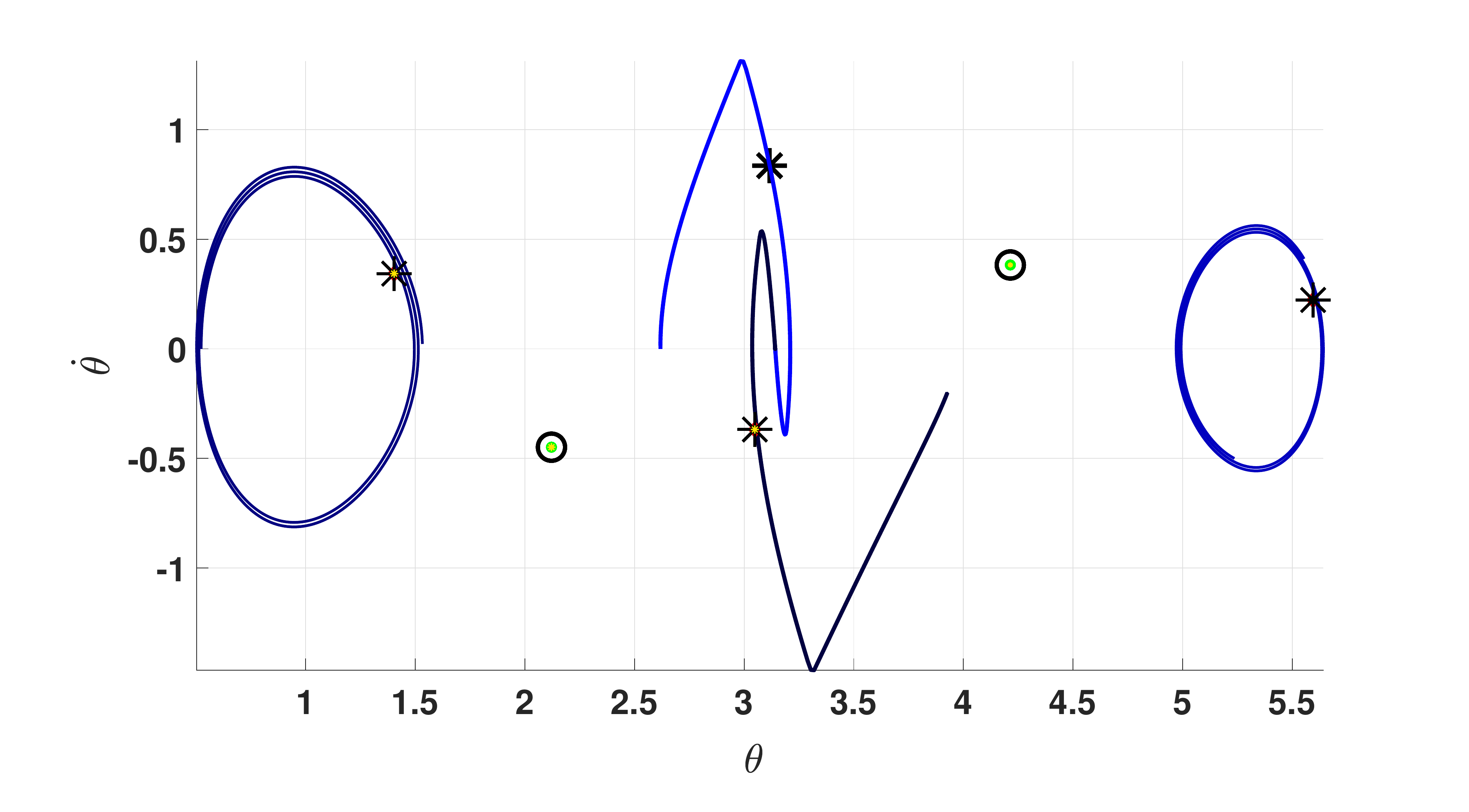}
    \caption{Data collection trajectories (solid lines) and query points ($\star$ and $\circ$) in state-space.}
    \label{fig:traj_data}
\end{figure}
Moreover, 
% following Theorem~\ref{thrm:AsympApprox}, 
the learned dynamics are more accurate for query points near the fixed point $\bar{x}_\textnormal{eq} = [\pi,0]^\top$ than for query points far away, as shown in Fig.~\ref{fig:traj_data}. For example, at query points $\bar{x} = \{[1.40\ 0.34]^\top,[3.05\ -0.37]^\top,[3.11\ 0.84]^\top,[5.60\ 0.22]^\top\}$, we see around $100\%$ decrease in the uncertainty range estimate as $N$ increases from $100$ to $4000$. The corresponding percentages for $\bar{x} = \{[4.21\ 0.38]^\top, [2.12\ -0.45]^\top\}$ are just around $73 \%$ and $34 \%$ respectively. }
% \begin{figure*}[h]
%     \centering
%     \includegraphics[scale=0.61]{SysIDfig.eps}
%     \caption{Scaled Uncertainty Estimates Over Query Points in the State-Space}
%     \label{fig:sysidpend}
% \end{figure*}

\subsection{Computation of Positive Invariant Set}
\textcolor{black}{For pendulum dynamics \eqref{eq:pend_ss} in discrete-time, we use the BMI \eqref{eq:s_proc_invar} of Proposition~\ref{prop:invsetLMI} to compute an ellipsoidal positive invariant set. In this specific example, we have used all the $N=4000$ samples from each of the previously collected four trajectories in Section~\ref{sysID_Pend}. The number of intersecting ellipsoids $n_I$ in \eqref{eq:s_proc_invar} is set as $n_I = 2$. Fig.~\ref{fig:invarsetsim} shows the invariant set $\mathcal{X}_I \subseteq \mathcal{X}$, where $\mathcal{X} = [0,2\pi]\times[-2.5,2.5]$. To further check numerically that the set $\mathcal{X}_I$ in Fig.~\ref{fig:invarsetsim} is indeed a positive invariant set, we run simulations from six initial conditions inside the set. As seen in Fig.~\ref{fig:invarsetsim}, all six trajectories stay within $\mathcal{X}_I$.}
\begin{figure}[t]
    \includegraphics[width = \columnwidth]{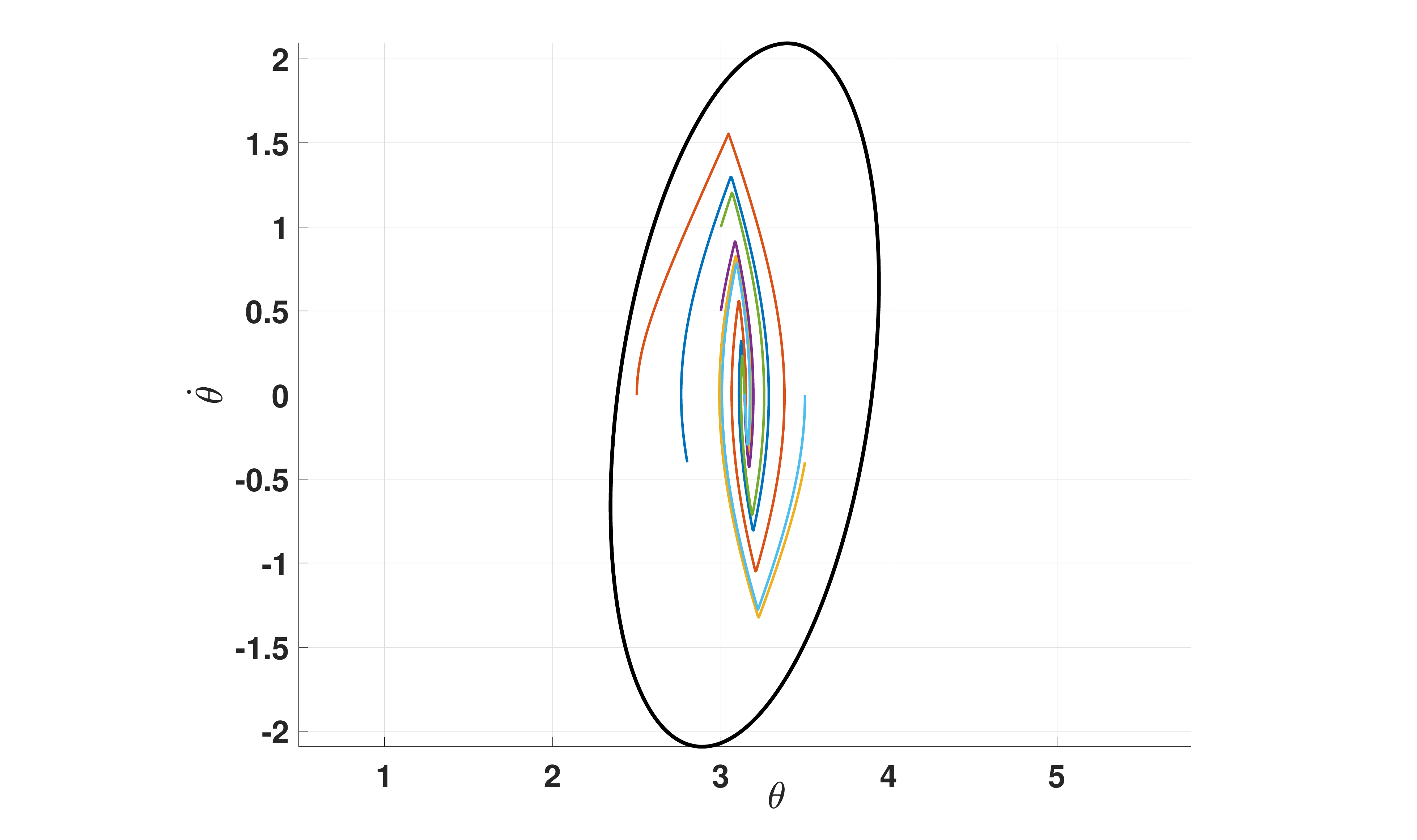}
    \caption{Invariant set (solid black line) computed using \eqref{eq:s_proc_invar}. Trajectories starting inside the set are contained within.}
    \label{fig:invarsetsim}
\end{figure}

\section{Conclusions}
We presented a non-parametric technique for online modeling of systems with nonlinear Lipschitz dynamics. The key idea is to successively use measurements to approximate the \emph{graph} of the function using envelopes described by quadratic constraints. Using techniques from convex optimization, we also computed a set valued estimate of the range of the unknown function at any given point in its domain, and a positive invariant set around a known equilibrium.
We further highlighted the efficacy of the proposed methodology via a detailed numerical example. 

\section*{Acknowledgements}
This work was partially funded by Office of Naval Research grant ONR-N00014-18-1-2833.
%%%%%%%%%%%%%%%%%%%%%%%%%%%%%%%%%%%%%%%%%%%%%%%%%%%%%%%%%%%%%%%%%%%%%%%%%%%%%%%%
\renewcommand{\baselinestretch}{0.81}
\bibliographystyle{IEEEtran}
\bibliography{IEEEabrv,bibliography} 
\appendix
\subsection{Tractable Optimization Problems for Section \ref{sec:app}}
\subsubsection{SDP for Ellipsoidal Outer Approximation of $\mathcal{F}(x_N)$}
Following \cite[Section~11.4]{calafiore2014optimization} finding the the minimum trace ellipsoid $\textnormal{ell}(c(x_N),R(x_N))$ that satisfies \eqref{eq:s_proc} can be posed as an SDP using Schur complement rule as:
\begin{equation*} 
    \begin{array}{llll}
        \displaystyle \min_{\xi} & \textnormal{trace}(R(x_N))\\
        ~ \text{s.t.} 
        & \begin{bmatrix}  p(x_N) & q(x_N) & -\mathbb{I}_n \\ q^\top(x_N) & r(x_N) & c^\top(x_N) \\ -\mathbb{I}_n & c(x_N) & -R(x_N) \end{bmatrix} \preceq 0, \\ 
        & \tau_k \geq 0,~\forall k = 0,1,\dots, N-1,\\
        & R(x_N) \succ 0,
    \end{array}
\end{equation*}
where $\xi = \{R(x_N), c(x_N), \tau_0,\dots, \tau_{N-1}\}$ and we have used the variable nomenclature \begin{align*}
    p(x_N) &= - \sum \limits_{k=0}^{N-1} \tau_k \mathbb{I}_n,~
    q(x_N) =   \sum \limits_{k=0}^{N-1} \tau_k f(x_k),\\
    r(x_N) &= -1 -\sum \limits_{k=0}^{N-1} \tau_k \Big (-L^2(x_N)^\top x_N + 2L^2 (x_N)^\top x_k \\ &~~~~~~~~~~~~~- L^2(x_k)^\top (x_k) + f^\top(x_k) f(x_k)\Big ).
\end{align*}

\subsubsection{Bisection Method for Positive Invariant Set Computation}
Note that \eqref{eq:s_proc_invar} is linear in $P_j$ for a fixed $\rho_{jm}$ and vice-versa. This facilitates using a bisection search on $\rho_{jm}$ until a feasible solution is obtained. For bounded $\mathcal{X}$, feasibility is guaranteed  for some $\rho$ such that all $\rho_{jm}=\rho$ because of continuity of \eqref{eq:s_proc_invar} as well as its feasibility for $\rho=0$.  After iterating over $\rho_{jm}$, \eqref{eq:s_proc_invar} is solved as a Linear Matrix Inequality (LMI)
\begin{equation*} 
    \begin{array}{llll}
        \displaystyle \min_{P_1,\dots, P_{n_I}, \tau_0,\dots, \tau_{N-1}} & \sum_{j=1}^{n_I}\textrm{trace}(P_j)\\
        \ \ ~~~~~~~~~ \text{s.t.} 
        & \eqref{eq:s_proc_invar}, \\ 
        & \tau_k \geq 0,~\forall k = 0,1,\dots, N-1.
    \end{array}
\end{equation*}
%\begin{spacing}{0.9}
%\bibliographystyle{IEEEtran}
%\bibliography{IEEEabrv, bibliography} 
%\end{spacing}
\subsection{Proofs}
\subsubsection*{Proof of Lemma \ref{lem:inter}}
For any $(x,f(x))\in G(f)$,  we have from the Lipschitz inequality,
\begin{align*}
    \Vert f(x)-f(y) \Vert\leq L\Vert x-y\Vert,\quad \forall y\in\mathcal{X},
\end{align*}
and choosing $y=x_k$ for $k=0,1,\dots N-1$ in the above inequality, in view of Corollary~\ref{corr:env} yields, 
\begin{align*}
    (x,f(x))&\in\mathcal{E}(x_k), \quad\forall k=0,1,\dots N-1,\\
    \Rightarrow (x,f(x))&\in\bigcap_{k=0}^{N-1} \mathcal{E}(x_k).
\end{align*}
Note the fact that $f$ is globally Lipschitz ensures that the intersections are non-empty.
Since this was shown for any $(x,f(x))\in G(f)$, we can thus conclude that 
$$ G(f)=\bigcup_{x\in\mathcal{X}}(x,f(x))\subseteq\bigcap_{k=0}^{N-1} \mathcal{E}(x_k). $$The other equalities follow from \eqref{eq:rer}.\hfill$\blacksquare$

\renewcommand{\baselinestretch}{0.85}
\subsubsection*{Proof of Proposition \ref{lem:rerfp}} 
We first prove the implication for when condition (1) holds. Let $k^\star$ be the time at which the system enters the $p$-periodic orbit, i.e., $\mathcal{O}_p(x_{k^\star})\subseteq\{x_0,x_1,x_2,\dots\}$. 
From Lemma~\ref{lem:inter} we have 
$\mathbf{E}^f_{N}=\bigcap_{k=0}^{N-1} \mathcal{E}(x_k)$ at any time $N$. For $N\geq k^\star+p$, consider the set
\begin{align*}
    \mathbf{E}^f_{N+1}&=\bigcap_{k=0}^{N} \mathcal{E}(x_k)\\
    &=(\bigcap_{k=0}^{k^\star-1}\mathcal{E}(x_k))\cap(\bigcap_{k=k^\star}^{k^\star+p-1}\mathcal{E}(x_k))\cap(\bigcap_{k=k^\star+p}^{N}\mathcal{E}(x_k)),\\
    &=\mathbf{E}^f_{k^\star+p}\cap(\bigcap_{k=k^\star+p}^{N}\mathcal{E}(x_k)).
\end{align*}
From Definition~\ref{def:porb}, we have that $x_k\in\mathcal{O}_p(x_{k^\star})$ for all $k=k^\star+p,\dots,N$. Using \eqref{eq:intenv} and the fact that $f(\cdot)$ is globally Lipschitz on $\mathcal{X}$, we thus have $\mathbf{E}^f_{k^\star+p}\subseteq \mathcal{E}(x_k)$, for all $k=k^\star+p,\dots,N$. Combining this implication with the definition of $\mathbf{E}^f_{N+1}$ above yields 
\begin{align*}
\mathbf{E}^f_{N+1}=\mathbf{E}^f_{k^\star+p},\quad\forall N\geq k^\star+p,
\end{align*}
and so $\mathbf{E}^f_{k^\star+p}$ is a fixed point for recursion \eqref{eq:rer}.\\
Now we prove the implication for when condition (2) holds, i.e., $\lim_{k\rightarrow\infty}x_k=\bar{x}_\textnormal{eq}$. Since the sets $\mathbf{E}^f_N=\bigcap_{k=0}^{N-1} \mathcal{E}(x_k)$ are non-increasing in the sense of \eqref{eq:envelorder}, the following limit set is well defined
\begin{align*}
\mathbf{E}^f_\star&=\lim_{N\rightarrow\infty}\mathbf{E}^f_N\\
&=\lim_{N\rightarrow\infty}\bigcap_{k=0}^{N-1} \mathcal{E}(x_k),\\
&=\lim_{N\rightarrow\infty}(\bigcap_{k=0}^{N-2} \mathcal{E}(x_k))\cap\mathcal{E}(x_{N-1}),\\
&=\lim_{N\rightarrow\infty}\bigcap_{k=0}^{N-2}\mathcal{E}(x_k)\cap  \mathcal{E}(\lim_{N\rightarrow\infty}x_{N-1}).
\end{align*}
The last equality follows from the property of product of convergent sequences because all the limits on both sides of the equation are well defined. Computing the limits then gives us the following equality
$$\mathbf{E}^f_\star=\mathbf{E}^f_\star\cap\mathcal{E}(\bar{x}_\textnormal{eq}). $$
Thus $\mathbf{E}^f_\star$ is a fixed point for recursion \eqref{eq:rer}.
% \balance
\hfill$\blacksquare$
% \balance
\subsubsection*{Proof of Theorem \ref{thrm:AsympApprox}}
From the definition $x_0$ in the theorem, we have that the sequence $\Big\{x_k\Big\}_{k=0}^\infty$ converges to the fixed point $\bar{x}_\textnormal{eq}$ of \eqref{eq:sys}. From the definition of the convergence of a sequence, we have that for every $\epsilon>0$, there exists a $\bar{N}(\epsilon)$, such that 
$$\Vert x_k-\bar{x}_\textnormal{eq} \Vert\leq\epsilon,\quad\forall k\geq \bar{N}(\epsilon).$$
The convergent sequence is a Cauchy sequence satisfying with the same $\epsilon$ and $\bar{N}(\epsilon)$ as above. That is,
\begin{align}\label{eq:appen1}
\Vert x_{k_1}-x_{k_2}\Vert \leq 2\epsilon, \quad \forall k_2,k_1 \geq \bar{N}(\epsilon).
\end{align}
Consider a subsequent queried point $\bar{x}$  at most $\epsilon$ distance away from $\bar{x}_\textnormal{eq}$. We further have from the Lipschitz inequality,
\begin{align}\label{eq:append2}
    \Vert f(x_{k_1})-f(\bar{x})\Vert &\leq L\Vert x_{k_1}-\bar{x} \Vert.
\end{align}
From Proposition~\ref{prop:ball_range}, we know that the possible values of $f(\bar{x})$ lie within a sphere of radius $L\Vert x_{k_1}-\bar{x} \Vert$ centered at $f(x_{k_1})$. The diameter of the above sphere bounds the maximum error in the estimate of $f(\bar{x})$, i.e.,
\begin{align*}
    &\Vert y - f(\bar{x})\Vert\leq 2L\Vert x_{k_1}-\bar{x} \Vert,\quad \forall y\in \mathcal{S}(x_{k_1},\bar{x}).
\end{align*}
For $k_1$ chosen as in \eqref{eq:appen1}, the above inequality can be written as 
\begin{align*}
    &\Vert y - f(\bar{x})\Vert\leq 4L\epsilon,\quad \forall y\in \mathcal{S}(x_{k_1},\bar{x}).
\end{align*}
Now for another $k_2$ chosen as in \eqref{eq:appen1} such that $f(x_{k_1})\neq f(x_{k_2})$, we have
\begin{align}\label{append3}
    \Vert f(x_{k_2})-f(\bar{x})\Vert &\leq L\Vert x_{k_2}-\bar{x}\Vert. 
\end{align}
The intersections of the envelopes constructed from \eqref{eq:append2} and \eqref{append3} is depicted in Fig.~\ref{fig:inter}.
\begin{figure}[h]
    \centering
    \includegraphics[width = 0.637\columnwidth]{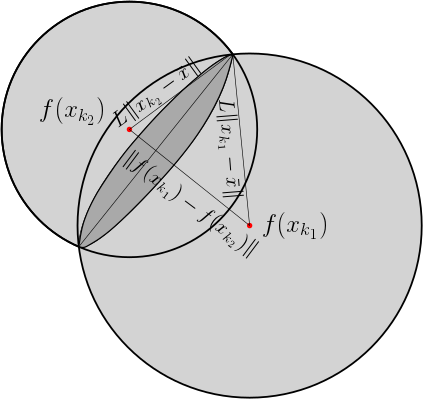}
    \caption{Intersection of set of possible values of $f(\bar{x})$, given two samples $x_{k_1}$, $x_{k_2}$}
    \label{fig:inter}
\end{figure}
We thus obtain a tighter bound on the error in the estimate of $f(\bar{x})$ via the diameter of the $n-2$ dimensional sphere obtained at the intersection of $n-1$ dimensional spheres, as given by
\begin{align*}
    &\Vert y - f(\bar{x})\Vert\leq\\ &2\sqrt{\frac{L^2(\Vert x_{k_1}-\bar{x} \Vert^2+\Vert x_{k_2}-\bar{x} \Vert^2)-\Vert f(x_{k_2})-f(x_{k_1})\Vert^2}{2}},\\
    &\leq\textrm{min}(2L\Vert x_{k_2}-\bar{x}\Vert, 2L\Vert x_{k_1}-\bar{x}\Vert)\leq 4L\epsilon,\\
    &\hfill\forall y\in \mathcal{S}(x_{k_1},\bar{x})\cap \mathcal{S}(x_{k_2},\bar{x})\\
    \Rightarrow &\max_{y\in\mathcal{S}(x_{k_1},\bar{x})\cap \mathcal{S}(x_{k_2},\bar{x})}\Vert y-f(\bar{x})\Vert\leq\\ &2\sqrt{\frac{L^2(\Vert x_{k_1}-\bar{x} \Vert^2+\Vert x_{k_2}-\bar{x} \Vert^2)-\Vert f(x_{k_2})-f(x_{k_1})\Vert^2}{2}},\\
    &\leq\textrm{min}(2L\Vert x_{k_2}-\bar{x}\Vert, 2L\Vert x_{k_1}-\bar{x}\Vert)\leq 4L\epsilon
\end{align*}
Taking intersections using all the envelopes collected (which are non-empty due to Lipschitz property of $f(\cdot)$ on $\mathcal{X}$) further shrinks the possible error and hence yields the desired result.

\hfill$\blacksquare$

\subsubsection*{Proof of Proposition \ref{prop:invsetLMI}} 
Consider any vector $[x^\top y^\top 1]^\top \in\mathbb{R}^{2n+1}$ such that $x\in\mathcal{X}_I$ and $[x^\top y^\top]^\top \in G(f)$. Given that BMI \eqref{eq:s_proc_invar} is feasible, we multiply $[x^\top y^\top 1]^\top \in\mathbb{R}^{2n+1}$ on both sides of \eqref{eq:s_proc_invar} for any $m\in\{1,2,\dots, n_I\}$ to get
\small
\begin{align*}
    &\sum_{j=1}^{n_I}\begin{bmatrix} -\rho_{jm}P_j& 0 & \rho_{jm}P_j\bar{x}_\textnormal{eq}\\0& P_m & -P_m\bar{x}_\textnormal{eq}\\\bar{x}^\top_\textnormal{eq}\rho_{jm}P_j&-\bar{x}^\top_\textnormal{eq}P_m&-\bar{x}^\top_\textnormal{eq}(\rho_{jm}P_j-P_m)\bar{x}_\textnormal{eq}+\rho_{jm}-1 \end{bmatrix}\\
    &~~~~~~~~~~~~~~~~~~~~~~~~~~~~~~~~~~~~~~~~~~~~~~~~~~~-\sum \limits_{k=0}^{N-1} \tau_k Q_L^f(x_k) \preceq 0,\\
    & \Rightarrow \sum_{j=1}^{n_I}-\rho_{jm}\begin{bmatrix}x\\1\end{bmatrix}^\top \bar{P}_j \begin{bmatrix}x\\1\end{bmatrix} + n_I\begin{bmatrix} y \\ 1 \end{bmatrix}^\top \bar{P}_m \begin{bmatrix}y\\1\end{bmatrix} \\ &~~~~~~~~~~~~~~~~~~~~~~~~~~~~~~~~~~~~-\begin{bmatrix}x \\ y \\ 1 \end{bmatrix}^\top \sum \limits_{k=0}^{N-1} \tau_k Q_L^f(x_k))\begin{bmatrix}x\\y\\1\end{bmatrix}\leq 0,\\
    & \Rightarrow \begin{bmatrix}y\\1\end{bmatrix}^\top \bar{P}_m \begin{bmatrix}y\\1\end{bmatrix}\leq 0,\ \forall m=1,2\dots, n_I,
\end{align*}
\normalsize
where $\bar{P}_j = \begin{bmatrix}P_j&-P_j\bar{x}_\textnormal{eq}\\-\bar{x}^{\top}_\textnormal{eq}P_j & \bar{x}^\top_\textnormal{eq}P_j\bar{x}_\textnormal{eq}-1\end{bmatrix}$. The last implication follows from the fact $x\in\mathcal{X}_I$ and $G(f)\subseteq \mathbf{E}^f$ (as proved in Lemma~\ref{lem:inter}). The last inequality further implies that $G(f)\ni y\in\mathcal{X}_I$, for all $x\in\mathcal{X}_I$, and hence, that $\mathcal{X}_I$ is invariant.
\hfill$\blacksquare$
%%%%%%%%%%%%%%%%%%%%%%%%%%%%%%%%%%%%

\end{document}